\documentclass[a4paper,12pt]{article}
\usepackage[T1]{fontenc}    
\usepackage{hyperref}       
\usepackage{url}            
\usepackage{booktabs}       
\usepackage{nicefrac}       
\usepackage{microtype}      

\usepackage{bm}
\usepackage{bbm}
\usepackage[utf8]{inputenc}
\usepackage[english]{babel}
\usepackage{amssymb}
\usepackage{amsmath}
\usepackage{amsthm}
\usepackage{amsfonts}
\usepackage{mathrsfs}
\usepackage{dsfont}
\usepackage{hyperref}
\usepackage{commath}
\usepackage{graphicx}
\usepackage{text comp}
\usepackage{mathtools}
\usepackage{afterpage}
\usepackage{appendix}
\usepackage[linesnumbered,lined,boxed,commentsnumbered]{algorithm2e}

\usepackage{xcolor}

\theoremstyle{plain}
\newtheorem*{thm*}{Theorem}
\newtheorem{thm}{Theorem}[section]

\newtheorem{lemma}[thm]{Lemma}
\newtheorem*{lemma*}{Lemma}
\newtheorem{corollary}[thm]{Corollary}
\newtheorem*{corollary*}{Corollary}

\newtheorem*{prop*}{Proposition}
\newtheorem{defn}{Definition}

\newtheorem*{conjecture*}{Conjecture}

\newcommand{\R}{\mathbb{R}}

\newcommand{\E}{\mathbb{E}}
\newcommand{\N}{\mathbb{N}}

\newcommand{\1}{\mathbbm{1}}

\newcommand{\emax}{e_{\max}}

\DeclarePairedDelimiter\floor{\lfloor}{\rfloor}

\everymath{\displaystyle}

\author{

        Desmond John Higham%
            \thanks{%
           School of Mathematics,
           University of Edinburgh,
           Edinburgh, EH9 3FD, UK
           (\texttt{d.j.higham@ed.ac.uk}).
         Both authors were supported by Engineering and Physical
     Sciences Research Council grant EP/P020720/1. 
        }
        
        \and
        
        Henry-Louis de Kergorlay%
    \thanks{%
           School of Mathematics,
           University of Edinburgh,
           Edinburgh, EH9 3FD, UK
           (\texttt{hdekerg@ed.ac.uk})
           }
\and 
        }
        
\date{}
\title{Epidemics on Hypergraphs: Spectral Thresholds for Extinction} 

\begin{document}

\maketitle
\begin{abstract}
    Epidemic spreading is well understood when a disease propagates around a contact graph. 
    In a stochastic susceptible-infected-susceptible setting, 
    spectral conditions characterise whether
    the disease vanishes. However, modelling human interactions using a graph is a simplification which only considers pairwise relationships.
    This does not fully represent the more realistic case where people meet in groups. Hyperedges can be used to record such group interactions, yielding more faithful and flexible models, allowing for the rate of infection of a node to vary as a nonlinear function of the number of infectious neighbors. We discuss different types of contagion models in this hypergraph setting, and 
    derive spectral conditions that characterize whether the disease vanishes. 
    We study both the exact individual-level stochastic model and a deterministic mean field ODE approximation.
    Numerical simulations are provided to illustrate the 
    analysis. We also interpret our results and show how the hypergraph 
    model allows us to distinguish between contributions to infectiousness that (a) are inherent in the nature of the pathogen and (b) arise from behavioural choices (such as social distancing, increased hygiene and use of masks). 
    This raises the possibility of more accurately quantifying the effect of interventions that are designed to contain the spread of a virus.   
\end{abstract}
\section{Introduction}

Compartmental models for disease propagation have a long and illustrious history
\cite{AM92,KM27}, and they 
remain a fundamental
predictive tool
\cite{EE2020,GBBCFPC20}.
For a stochastic, individual-level model, it has been suggested recently that 
hyperedge information
should be incorporated  \cite{APM20,simplicialSocialContagion,heterogeneityHypergraph}.
Hyperedges allow us to 
account directly for 
group interactions of any size,
rather than, as in, for example,  
\cite{epidemicsSpread,virusSpreadInNetworks,spreadingEigenvalue}, 
treating them as 
a collection of essentially independent 
pairwise encounters.

In this work we
contribute to the 
modeling and analysis 
of  disease spreading on a hypergraph.
We include the case where the number of infected nodes in a hyperedge contributes nonlinearly to the overall infection rate; this covers the 
so-called \emph{collective contagion model}
setting and a new alternative 
that we call a 
 \emph{collective suppression model}.
The main contributions of our work are 
\begin{itemize}
    \item a mean field approximation (\ref{dynamical system})--(\ref{eq:gidef}) with a spectral condition for local asymptotic stability of the zero-infection state 
    (Theorem~\ref{thm: local asymptotic stability}) 
    and 
     an extension to global asymptotic stability
     (Theorem~\ref{thm: global asymptotic condition})
     when the nonlinear infection function is concave,
    \item for the exact, individual-level model, a spectral condition for exponential decay
    of the non-extinction probability
    in the concave case (Theorem~\ref{thm: exponential decay}) and a spectral bound on the expected time to extinction (Corollary~\ref{cor:et}),
    \item extensions of these results to more general
    partitioned hypergraph models, where distinct infection rates apply to different categories of hyperedge (\ref{eq:partition}),
    \item results for the non-concave collective 
    contagion model (Theorem~\ref{thm: collective contagion model} and 
    Theorem~\ref{thm: max collective contagion model}),
    \item a complementary condition that rules out extinction of the disease (Theorem~\ref{thm: non-zero steady-state}), 
    \item interpretations of these mathematical results: the spectral
    thresholds for disease extinction naturally distinguish between 
    the 
    inherent biological infectiousness of the disease and behavioural choices of the individuals in the population, allowing us to account for intervention strategies 
    (Section~\ref{sec:disc}).
\end{itemize}

The manuscript is organized as follows.
In section~\ref{sec:SIS}
we introduce the traditional 
graph-based susceptible-infected-susceptible 
(SIS) model and quote a spectral condition 
that
characterizes
control of the disease.
We then 
discuss the generalization to hyperedges, and
motivate the 
 use of infection rates that 
 do not scale linearly with 
 respect to the number of infected neighbors.
 Section~\ref{sec:hyp_model}
 formalizes the hypergraph 
 model and 
 shows how it may be simulated.
 In section~\ref{sec:mean_field}
 we derive a mean field approximation
 to characterize the behaviour of the model, and 
 in section~\ref{sec:comp} we 
 give computational results
 to illustrate its relevance.
 Section~\ref{sec:stab}
 analyses the deterministic mean field setting and 
 gives a spectral condition for
 long-term decay of the disease.
 The result is local for general infection rates and global (independent of the initial condition)
 for the concave case.
 This spectral condition generalizes a well-known result concerning disease 
 propagation on a graph.
 Section~\ref{sec:more_sim} provides further computational simulations to illustrate the spectral threshold. 
 The full stochastic model is then studied in 
 section~\ref{sec:exact_model},
 where we extend 
 the analysis in \cite{epidemicsSpread} to our hypergraph setting.
 Here we study extinction of the disease
 in the case where 
 the nonlinearity in the infection rate is concave.
 We also derive conditions for non-extinction.  
 In section~\ref{sec:cim}
 we extend 
 our analysis to the 
 so-called
 collective contagion model
 proposed in 
 \cite{APM20,simplicialSocialContagion,heterogeneityHypergraph}.
 In section~\ref{sec:disc}
 we summarize and interpret our results, and 
 discuss follow-on work.

 For a review of recent studies of spreading processes on hypergraphs, including 
 the dissemination of rumours, opinions and knowledge, 
 we recommend \cite[subsection~7.1.2]{BCILLPYP21}.
 The model that we study fits into the framework of 
 \cite{SISonHypergraphs}. This work introduced the idea 
of a nonlinear ``infection pressure''  from each hyperedge, and derived a mean field approximation that was compared with microscale-level simulation results.
In \cite{simplicialSocialContagion}, the authors studied
this type of model on simplicial complexes of degree up to two (a subclass of the more general hypergraph setting) and also studied a mean field approximation.
These authors examined the mean field system from a dynamical systems perspective and analysed issues such
as bistability, hysteresis and discontinuous transitions.
Similarly, in \cite{APM20,heterogeneityHypergraph}, a hypergraph version was considered. Our work differs from these studies in 
(a) focusing on the derivation of spectral thresholds 
for extinction of a disease in both the exact and mean field settings and (b) 
seeking to interpret the results from a 
mathematical modelling perspective.
We mention that it would also be of interest to 
develop corresponding thresholds for the  
mean field models in
\cite{SISonHypergraphs,APM20,simplicialSocialContagion}.

\section{Stochastic SIS models}
\label{sec:SIS}

\subsection{Stochastic SIS model on a graph}
\label{subsec:graph}
Classical ODE compartmental models are based on the assumption that 
any pair of individuals is equally likely to interact---this 
is the homogeneous mixing case \cite{AM92}.
If, instead, we have knowledge of all 
possible pairwise interactions between individuals,
then this information may be 
incorporated via a contact graph and used in a stochastic model.
Here, each node represents an individual, and 
an edge between nodes $i$ and $j$ indicates that 
individuals $i$ and $j$ interact.
For a population with $n$ individuals, 
we may let 
$A \in \mathbb{R}^{n \times n}$ denote the 
corresponding symmetric adjacency matrix, so nodes $i$ and $j$ interact if and only
$A_{ij} = 1$.
In this setting, a stochastic SIS model 
uses the two-state random variable
$X_i(t)$
to represent the §
status of node $i$ at time $t$, with $X_i(t) = 0$
for a susceptible node and 
$X_i(t) = 1$
for an infected node. Each 
$X_i(t)$
then 
follows
a continuous time Markov process where 
the infection rate is given by 
\begin{equation}
    \beta \sum_{j=1}^n A_{ij} X_j(t)
    \label{eq:graph_inf}
\end{equation}
and the recovery rate is $\delta$. Here,
$\beta > 0$ and $\delta > 0$ are parameters governing the strength of the two effects. 
In this model, we see from (\ref{eq:graph_inf}) that the 
current 
chance of infection increases linearly in proportion to the
current number of infected neighbors.

This model was studied in 
\cite[Theorem~1]{spreadingEigenvalue}, where it was argued that the  
condition
\begin{equation}
\lambda(A)    \frac{\beta}{\delta} <1
\label{eq:graph_epi}
\end{equation}
guarantees the disease will die out.
Here, $\lambda(A)$ denotes the largest eigenvalue of the symmetric matrix $A$.
Further justification for this result may be found, for example,
in 
\cite{epidemicsSpread,virusSpreadInNetworks}. 
We note that (\ref{eq:graph_epi}) gives an elegant generalization of the homogeneous mixing case (where $A$ corresponds to the complete graph).

\subsection{Why use a hypergraph?}
\label{subsec:hypergraph}
It has been argued 
\cite{ABAMPL21,BCILLPYP21,BASJK18,benson2016higher,ER06}
that 
in many network science applications 
we lose information by 
recording only pairwise interactions. 
For example, emails can be sent to groups of recipients,
scholarly articles may have multiple coauthors,
and
many proteins may interact to form a complex.  
In such cases, recording the relevant lists of interacting 
nodes gives a more informative picture than reducing 
these down to a collection of edges.

In the setting of an SIS model, we may argue that
individuals typically come together in
well-defined groups, for example, in a household, 
a workplace or a social setting.
Such groups may be handled by the use of hyperedges,
leading to a hypergraph; these 
concepts are formalized in the next section.

With a classic graph model, as described in section~\ref{subsec:graph},
the rate of infection of a node is linearly proportional to the number of infectious neighbors.
With a hypergraph we may consider more intricate contagion mechanisms.
For example, 
using the terminology of 
\cite{heterogeneityHypergraph}, 
the \textit{collective contagion} model 
is used in \cite{APM20,simplicialSocialContagion,heterogeneityHypergraph}). Here, infection only starts spreading within a hyperedge after a certain threshold number of infectious neighbors has been reached. 
This type of behaviour is relevant, for example, in an office environment. A small number of workers may be able to 
socially distance in way that effectively eliminates the 
risk of infection.
However,
if the number of 
individuals (size of the hyperedge) is too large, then 
the disease may spread.

We mention here that an alternative type of mechanism may also  
operate, which we call \textit{collective suppression}.
Imagine that a disease may be contracted through 
contact with a surface that was previously touched by an infected individual.
Now suppose that 
a group of individuals is likely to use the 
same physical object, such as a door handle, hand rail,
cash machine, or water cooler.
If an infected individual contaminates the object, then further 
contamination by other individuals is less relevant.
In this case,
doubling the number of common users will increase the 
risk of infection by a factor less than two; generally 
risk grows sublinearly as a function of the size of the hyperedge. 

These arguments motivate us to study the case where 
the rate of infection of a node within a hyperdege is dependent on a generic function $f$ of the number of infectious neighbors in a hyperedge; this approach was also taken in  \cite{SISonHypergraphs}.
We will be particularly concerned with the 
case where $f$ is concave, since this 
is tractable for analysis and allows us to 
draw conclusions about the 
collective contagion model.

We note that if $f$ is the identity, then we recover linear dependence on the number of infectious neighbors and the hypergraph model is equivalent to a virus spreading on the clique graph of the hypergraph.

\section{SIS on a Hypergraph}

\label{sec:hyp_model}
\subsection{Background}
We continue with some standard definitions
\cite{bretto2013hypergraph}. 

\begin{defn}
A \emph{hypergraph} is a tuple $\mathcal H:=(V,E)$ of \emph{nodes} $V$ and \emph{hyperedges} $E$ such that $E\subset {\mathcal P}(V)$.
Here, ${\mathcal P}(V)$ denotes the power set of $V$. 
 \end{defn}
 
 We will let $n$  and $m$ denote the number of nodes
 and hyperedges, respectively; that is, 
 $|V| =n$ and $|E| = m$. Loosely, a hypergraph generalizes 
the concept of a graph by allowing an ``edge'' to be a list of more than two nodes. 

\begin{defn}
Consider a hypergraph $\mathcal H:=(V,E)$.
The \emph{incidence matrix},
${\mathcal I}$, 
is the $n\times m$ matrix such that  
$\mathcal I_{i h}=1$ if node $i$ belongs to hyperedge $h$ and $\mathcal I_{i h}=0$ otherwise.
\end{defn}

It is also useful to introduce 
$W:={\mathcal I} {\mathcal I}^T$.
This $n\times n$ matrix has the property that 
$W_{i j}$ records the number of hyperedges containing both nodes $i$ and $j$.
In particular, if $\mathcal H$ is a graph then $W$ is the affinity matrix of the graph.

\subsection{General infection model}
\label{subsec:model}

In our context, the nodes represent individuals and 
a
hyperedge records a collection of individuals who 
are known to interact as a group.
As in the graph case introduced in subsection~\ref{subsec:graph},
we use a state vector $X(t)$
which follows a continuous time Markov process, 
where, for each $1 \le i \le n$,
$X_i(t) = 1$ if node $i$ is infected at time 
$t$ and
$X_i(t) = 0$ otherwise.
We continue to assume that 
an infectious node becomes susceptible with 
constant recovery rate $\delta>0$.
However, 
generalizing 
(\ref{eq:graph_inf}), we now assume 
that a susceptible node $i$ becomes infectious with
rate 
\begin{equation}\label{infection rate random}
\beta\sum_{h\in E}\mathcal I_{i h}f\big(\sum_{j = 1}^{n} X_j(t)\mathcal I_{j h}\big),
\end{equation}
where $\beta>0$ is a constant. 
In (\ref{infection rate random}), 
$f:\R_+\to \R_+$
specifies the 
manner in which
the contribution to the overall level of 
infectiousness 
from each hyperedge involving node $i$ 
is assumed to increase 
in proportion to the number of 
infected nodes in that hyperedge.
Throughout our analysis we will always assume that 
$f(0)=0$ and 
$f$ is $C^1$ in a neighborhood of $0$.

 If $f$ is the identity, the rate of infection reduces to $\beta \sum_{j=1}^{n} W_{i j}X_j(t)$. This gives a weighted version of the infection rate of an SIS model on a graph.
 As discussed in subsection~\ref{subsec:hypergraph},
 it may be appropriate to choose 
 nonlinear $f$ in certain circumstances.
  We note that in \cite{SISonHypergraphs} the authors have in mind functions which behave like the identity near the origin and have a horizontal asymptote. Instances of such functions are 
  $x\mapsto \mathrm{arctan}(x)$ and $x\mapsto \min\{x,c\}$ for some $c>0$. Relaxing these conditions, we may ask more generally in such a
  setting that the function be concave. On the other hand, the authors in \cite{APM20,simplicialSocialContagion} consider a \emph{collective contagion model}, where infection spreads within a hyperdge only if a certain threshold of infectious vertices is reached in that hyperedge.
   A collective contagion model may be represented via the function $x\mapsto c_2\1(x\geq c_1)$ for some $c_1,c_2>0$, or $x\mapsto \max\{0,x-c\}$ for some $c>0$.
 
 \subsection{Partitioned hypergraph model}
 \label{subsec:partitoned}
 We also introduce a more general case where we partition the hyperedges into $K$ disjoint categories with each category $1 \le k \le K$ having its own distinct rate of infection in response to the number infected nodes in a hyperedge, represented by a function $f_k$.
 For example, the categories may correspond to different types of housing, workplaces, hospitality venues or sports facilities. 
 We may then represent the infection rate of node $i$ as
\begin{equation}\label{eq:partition}
\beta \sum_{k=1}^{K}\sum_{h\in E}{\mathcal I}^{(k)}_{i h}f_k\left(\sum_{j=1}^{n}X_j(t){\mathcal I}^{(k)}_{j h}\right),
\end{equation}
where we let ${\mathcal I}^{(k)}_{i h}=1$ if $i$ belongs to hyperedge $h$ in the category $k$ and ${\mathcal I}^{(k)}_{i h}=0$ otherwise; so  
${\mathcal I}^{(k)}$ is the incidence matrix of the subhypergraph consisting of only the hyperedges from category $k$.
We will refer to this as a 
\emph{partitioned hypergraph model}.

In this generalized case, a collective contagion model could be defined by first organizing the hyperedges into categories depending on their size, so that category $k$ is the set of hyperedges of size $k+1$. A collective contagion model may then represented, for example, via the functions $f_1:x\mapsto x$, and $f_k:x\mapsto c_{2,k}\1(x\geq c_{1,k})$, $k\in \{2,\dots,K\}$.

\section{Mean Field Approximation}
\label{sec:mean_field}


A classic approach to studying processes such as 
(\ref{infection rate random}), where 
infection rates are random, 
is to develop 
a mean field approximation for the expected process $$\big(\E[X_i(t)]\big)_{t\geq 0}=\big(\mathbb P(X_i(t)=1)\big)_{t\geq 0}=:(p_i(t))_{t\geq 0},$$
with 
deterministic rates.
In our case, 
the rate of recovery $\delta$ is constant, so can remain unchanged. Let us express the rate of infection (\ref{infection rate random}) of a node solely in terms of the expected processes 
$\{p_i(t)\}_{i=1}^{n}$. To do this we can substitute the $X_j(t)$ appearing in (\ref{infection rate random}) by their expected values $p_j(t)$. The approximate rate of infection for node $i$ then becomes
\begin{equation}\label{infection rate approximate}
\beta \sum_{h\in E}\mathcal I_{i h}f\big(\sum_{j=1}^{n} p_j(t)\mathcal I_{j h}\big).
\end{equation}
We arrive at the deterministic mean field ODE
\begin{equation}\label{dynamical system}
\frac{d P(t)}{dt}=g(P(t)),
\end{equation}
where $g_i:\R^n\to \R$ is defined by
\begin{equation}
g_i(P(t)):=\beta\sum_{h\in E}\mathcal I_{i h}f\left(\sum_{j=1}^{n}p_j(t)\mathcal I_{j h}\right)(1-p_i(t))-\delta p_i(t).
\label{eq:gidef}
\end{equation}

\section{Simulations and Comparison between Exact and Mean Field Models}
\label{sec:comp}

Let us emphasize that the approximate infection rates in (\ref{infection rate approximate}) differ 
in general from the expectation of the random rates in (\ref{infection rate random}). When the function $f$ is concave, however, Jensen's reverse inequality indicates that the rates in (\ref{infection rate approximate}) are greater than the expectation of the rates in (\ref{infection rate random}). Hence, in this case the expected quantities $p_i(t)$ are overestimated by 
(\ref{dynamical system})--(\ref{eq:gidef}). 
This is fine since we are looking for conditions for the disease to vanish. If $f$ is not concave (e.g., for a collective contagion model), these expected quantities are underestimated and it is not clear a priori whether the exact model is well approximated by the mean field ODE. 

In this section we therefore present 
results of computational simulations in order to gain insight into the accuracy of our mean field
approximation.

\subsection{Simulation algorithm}

Before presenting numerical results, we 
summarize our approach for simulating the 
individual-level stochastic model, which is based 
on a standard 
time discretization;
see, for example, \cite{SISonHypergraphs}. 
Using a small fixed time step $\Delta t$, we advance from time $t$ to $t+\Delta t$ as follows.
First, let $r\in [0,1]^n$ be a random vector 
of i.i.d.\ values uniformly sampled from $[0,1]$. For every node $1 \le i \le n$,
\begin{itemize}
    \item when $X_i(t)=0$, set $X_i(t+\Delta t)=1$ if
$$
r_i<1-\exp\left(-\beta \sum_h\mathcal I_{i h}f(\sum_j X_j(t)\mathcal I_{j h})\Delta t\right),
$$
and set $X_i(t+\Delta t)=0$ otherwise; 
\item when $X_i(t)=1$, set $X_i(t+\Delta t)=0$ if
$$
r_i<1-\exp\left(-\delta\Delta t\right),
$$
and set $X_i(t+\Delta t)=1$ otherwise.
\end{itemize}

\subsection{Computational results}
\label{subsec:comp}


In the simulations we chose $n=400$ nodes with fixed
recovery rate $\delta =1$.
We look at results for different choices of 
infection strength $\beta$ and $i_0$, the latter denoting the (independent) initial probability for each node to be infectious. 
We simulated the mean field ODE using Euler's method 
with time step $\Delta t=0.05$. The largest size of a hyperedge was $5$ and we distributed the number of hyperedges for the hypergraph randomly as follows: $300$ edges, $200$ hyperedges of size $3$, $100$ hyperedges of size $4$ and $50$ hyperedges of size $5$. To give a feel for the level of fluctuation, the individual-level paths are averaged over $10$ runs, each 
with the same hypergraph connectivity and initial state.

Figures~\ref{fig: min},
\ref{fig: log} and 
\ref{fig: atan}
show results for three concave choices of $f$; respectively,
\begin{itemize}
    \item $f(x)=\min(x,3)$,
    \item $f(x) = \log(1+x)$,
    \item $f(x) = \mathrm{arctan}(x)$.
\end{itemize}
For 
Figure~\ref{fig: max}
we used a collective contagion model on a partitioned
hypergraph. Assigning each hyperedge to a category in $\{1,2,3,4\}$, where category $k$ contains the hyperedges of size $k+1$, we chose the following associated functions to determine the infection rates: $f_1(x):=x$, and for $k\in \{2,3,4\}$, $f_k(x):=(k-1)\1(x\geq k-1)$.

The four figures show the proportion of infectious individuals as a function of time. 
In these simulations, and others not reported here,
we observe that the initial value $i_0$ does not affect the asymptotic behavior of the process: the process vanishes or converges to a non-zero equilibrium depending on the value of $\beta$ but regardless of the value of $i_0$.
In 
Figures~\ref{fig: min},
\ref{fig: log} and 
\ref{fig: atan}, where  $f$ is concave, we know that 
the mean field model gives an upper bound
on the expected proportion of infected individuals
in the microscale model. We also see
that the mean field model provides a reasonably sharp approximation. 
Moreover, we see a similar level of sharpness in  
Figure~\ref{fig: max} for the 
collective contagion model, where 
$f$ is not concave.


\begin{figure}[htp]
    \centering
    \includegraphics[width=\textwidth]{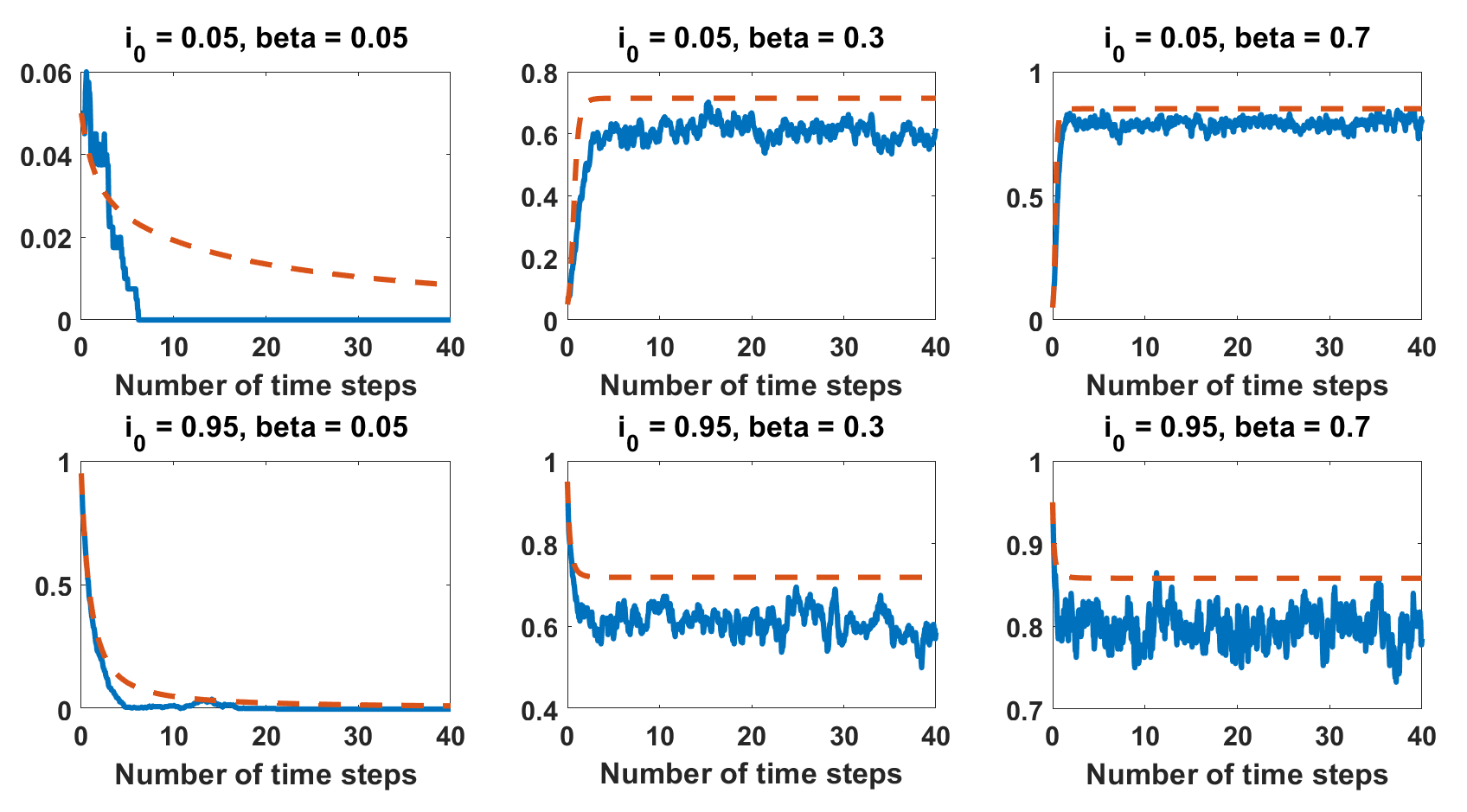}
    \caption{Here, $f(x) = \min(x,3)$.  
     Red dashed line: mean field approximation from 
   (\ref{dynamical system})--(\ref{eq:gidef}).
    Blue solid line:
     proportion of infected individuals,
     $\sum_{i} X_i(t)/n$, from the individual-level stochastic model 
     (\ref{infection rate random}), averaged over $10$ runs.
    }
    \label{fig: min}
\end{figure}

\begin{figure}[htp]
    \centering
    \includegraphics[width=\textwidth]{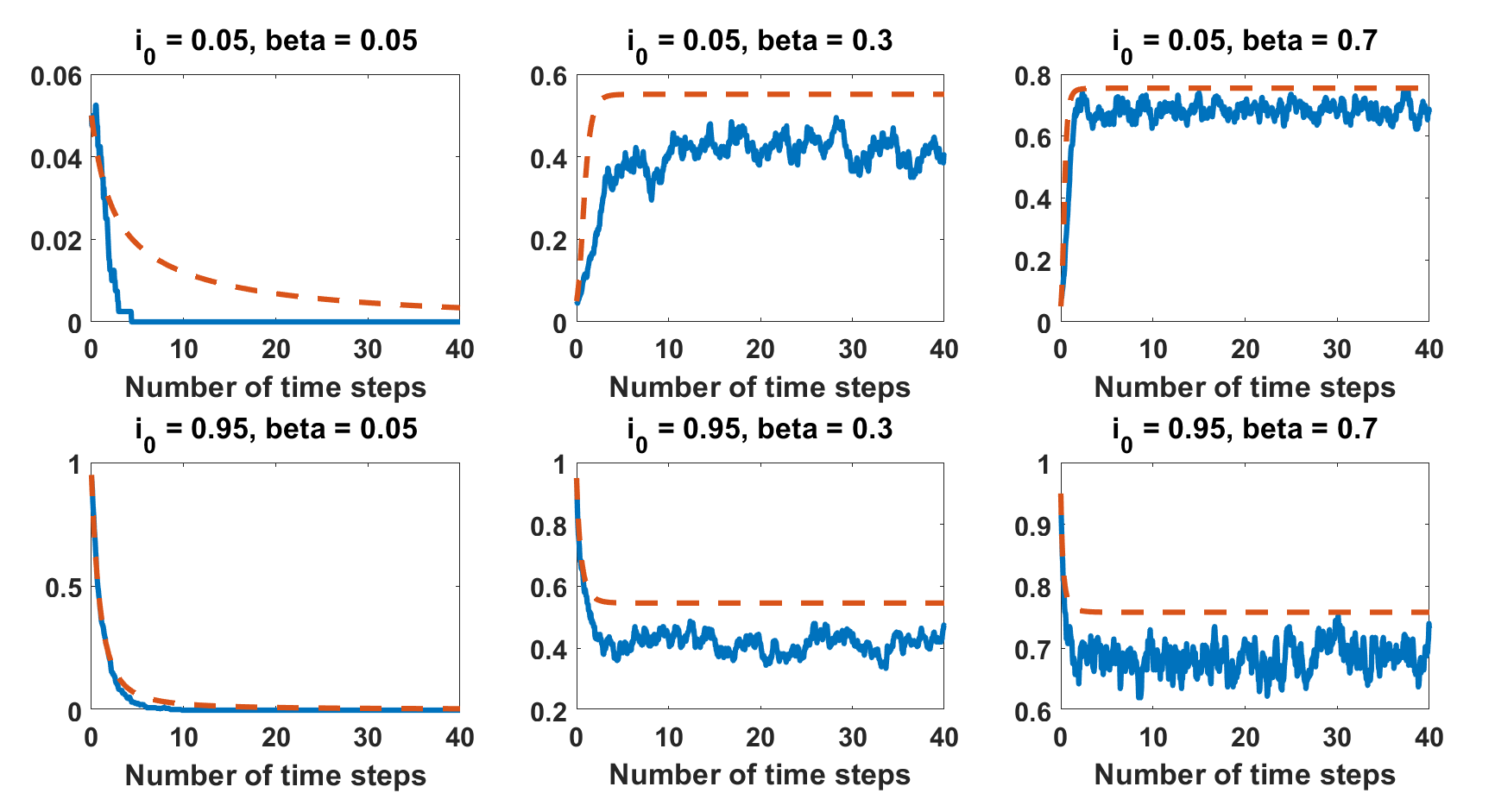}
    \caption{Here, $f(x) = \log(1+x)$.  
     Red dashed line: mean field approximation from 
   (\ref{dynamical system})--(\ref{eq:gidef}).
    Blue solid line:
     proportion of infected individuals,
     $\sum_{i} X_i(t)/n$, from the individual-level stochastic model 
     (\ref{infection rate random}), averaged over $10$ runs.
    }
    \label{fig: log}
\end{figure}

\begin{figure}[htp]\label{fig: atan}
    \centering
    \includegraphics[width=\textwidth]{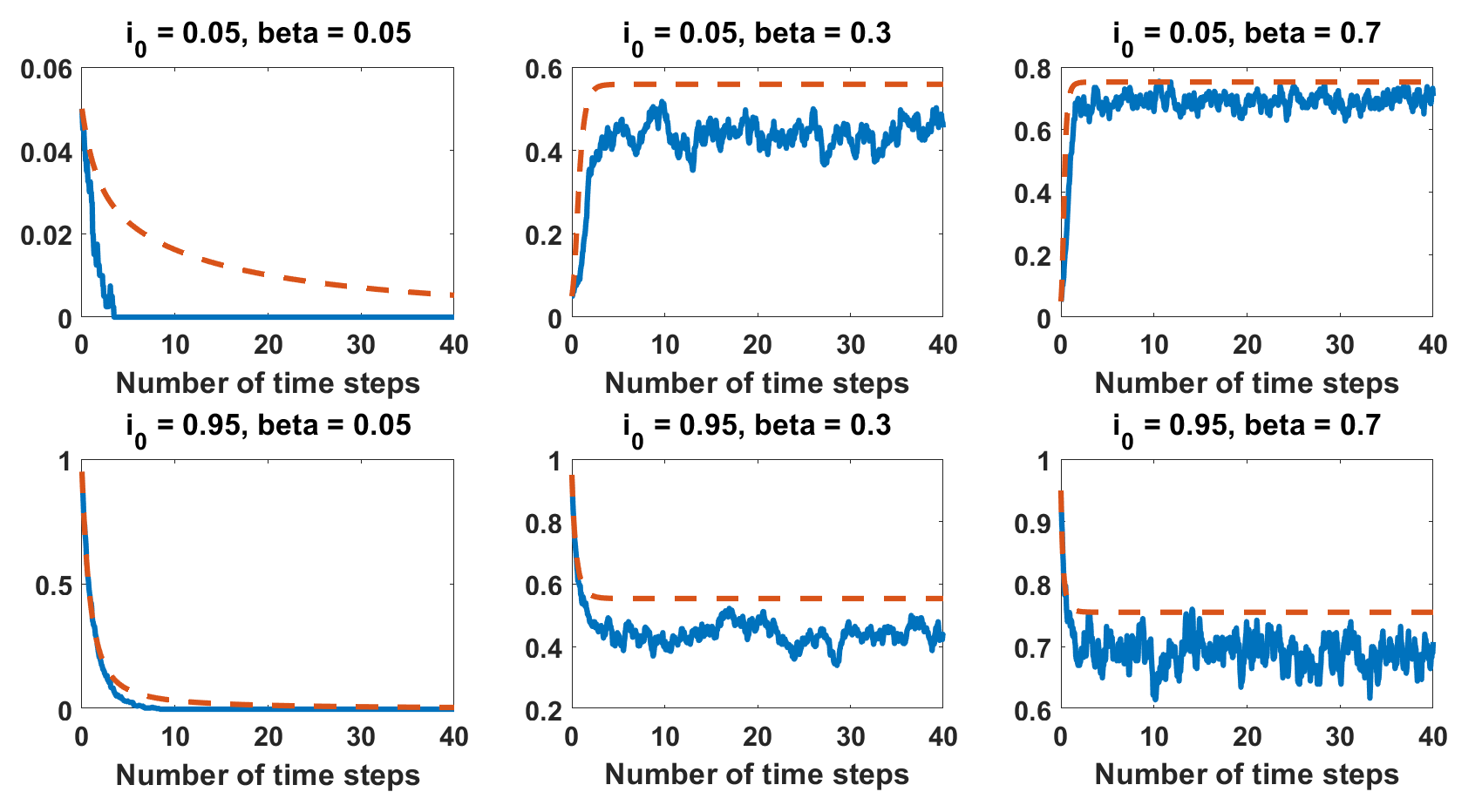}
    \caption{Here, $f(x) = \mathrm{arctan}(x)$.  
     Red dashed line: mean field approximation from 
   (\ref{dynamical system})--(\ref{eq:gidef}).
    Blue solid line:
     proportion of infected individuals,
     $\sum_{i} X_i(t)/n$, from the individual-level stochastic model 
     (\ref{infection rate random}), averaged over $10$ runs.
    }
    \label{fig: atan}
\end{figure}

\begin{figure}[htp]
    \centering
    \includegraphics[width=\textwidth]{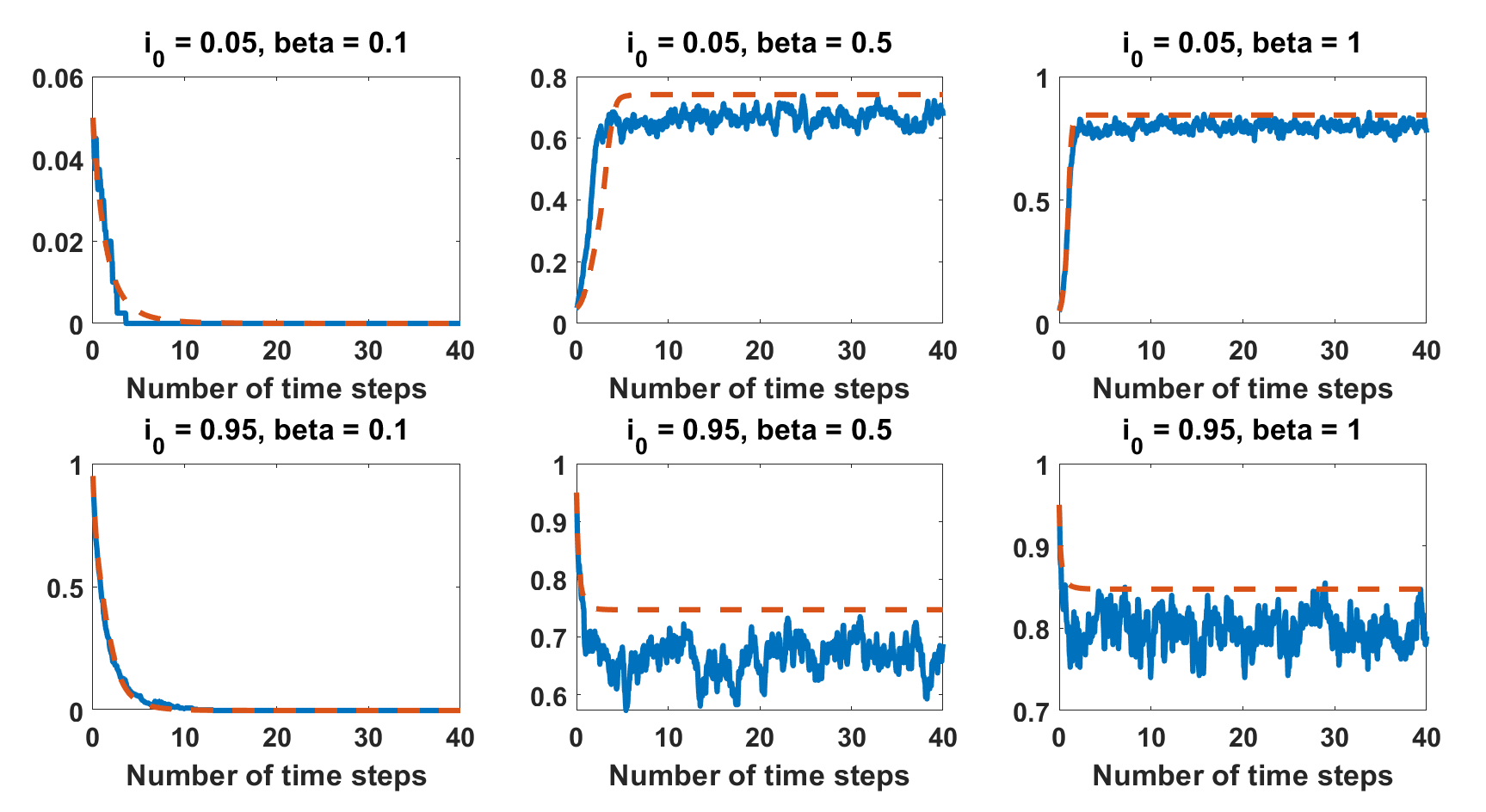}
    \caption{Collective contagion model on a partitioned hypergraph.
    Red dashed line: mean field approximation from 
   (\ref{dynamical system}) with (\ref{eq:mfpart}).
    Blue solid line:
     proportion of infected individuals,
     $\sum_{i} X_i(t)/n$, from the individual-level stochastic model 
     (\ref{infection rate random}), averaged over $10$ runs.
    }
    \label{fig: max}
\end{figure}


A key advantage of the mean field approximation is that it gives rise to a deterministic autonomous dynamical system for which there exists a rich theory to study the asymptotic stability of equilibrium points.
This motivates the analysis in the next section.

\section{Stability Analysis}
\label{sec:stab}

We provide below spectral conditions which imply that 
the infection-free solution $0 \in \R^n$ is a locally or globally asymptotically stable equilibrium of 
(\ref{dynamical system})--(\ref{eq:gidef}).
We will find that local asymptotic stability can be shown with no structural assumptions on $f$. 
We will also find that
global asymptotic stability
follows \emph{under the same conditions} when $f$ is concave. 
 Our conclusions fit into a framework that generalizes 
the graph case (\ref{eq:graph_epi}):
the spectral threshold takes the form 
\[
\lambda(W)\frac{c_f \beta}{\delta} <1 
\]
for some constant $c_f>0$ depending only on the choice of $f$.

Throughout this work, to be concrete we 
let $\| \cdot \|$ denote the Euclidean norm.

\subsection{Local asymptotic stability}
\begin{thm}\label{thm: local asymptotic stability}
If
\begin{equation}
\lambda(W)
\frac{ f'(0) \beta}{\delta} < 1
\label{eq:W1}
\end{equation}
then 
$0 \in \R^n$ is a locally asymptotically stable equilibrium of (\ref{dynamical system})--(\ref{eq:gidef}); that is, 
there exists a positive $\gamma$ such that 
 $||P(0)||<\gamma \Rightarrow\lim_{t\to\infty}||P(t)||=0$.
\end{thm}

\begin{proof}
We see that $g(0)=0$, so $0 \in \R^n$ is an equilibrium for (\ref{dynamical system}). It remains to show that this solution is locally asymptotically stable. 
Appealing to a standard linearization 
result~\cite{Ve90}, it suffices to show that every eigenvalue of the Jacobian matrix $\nabla g(0)$ has a negative real part.
We compute
\begin{align*}
    \frac{\partial g_i}{\partial p_{j_0}}&=\begin{cases}
    \beta\sum_h\mathcal I_{i h}\mathcal I_{j_0 h}f'(\sum_j p_j\mathcal I_{jh})(1-p_i),&\text{ }j_0\neq i,\\
    \beta\sum_h\mathcal I_{i h}\mathcal I_{j_0 h}f'(\sum_j p_j\mathcal I_{jh})(1-p_i)-\beta\sum_{h}\mathcal I_{ih}f(\sum_j p_j\mathcal I_{j h})-\delta,&\text{ }j_0=i.
    \end{cases}
\end{align*}
We see that $ \nabla g(0)= \beta f'(0)W-\delta I$.
This matrix is symmetric and therefore has real eigenvalues.
Hence, 
it suffices that the 
largest eigenvalue of 
$\beta f'(0)W$ does not exceed
$\delta$, and the result follows.
\end{proof}

Theorem~\ref{thm: local asymptotic stability} extends  to the partitioned 
model in 
(\ref{eq:partition}).
In this case  $g_i(P(t))$ in the mean field ODE 
(\ref{dynamical system}) is defined as
\begin{equation}
g_i(P(t)):=\beta\sum_{k=1}^{K}\sum_{h\in E}\mathcal I_{i h}^{(k)}f_k\big(\sum_{j=1}^{n}p_j(t)\mathcal I_{j h}^{(k)}\big)(1-p_i(t))-\delta p_i(t),
\label{eq:mfpart}
\end{equation}
and we let 
 $W^{(k)}:= {\mathcal I}^{(k)} {{\mathcal I}^{(k)}}^T$. 

\begin{thm}\label{thm: generalized local asymptotic stability}
If
\begin{equation}
\lambda\left(\sum_{k=1}^{K} f_k'(0)W^{(k)}\right)\frac{\beta}{\delta} < 1 
\label{eq:Wk}
\end{equation}
then
$0 \in \R^n$ is a locally asymptotically stable equilibrium 
of (\ref{dynamical system}),
with $g_i$ defined in (\ref{eq:mfpart}).
\end{thm}
\begin{proof}
The proof of 
Theorem~\ref{thm: local asymptotic stability}
extends straightforwardly.
We compute
\begin{align*}
    \frac{\partial g_i}{\partial p_{j_0}}&=\begin{cases}
    \beta\sum_{k}\sum_h\mathcal I^{(k)}_{i h}\mathcal I^{(k)}_{j_0 h}f_k'(\sum_j p_j\mathcal I^{(k)}_{jh})(1-p_i),&\text{ }j_0\neq i,\\
    \beta\sum_{k}\sum_h\mathcal I^{(k)}_{i h}\mathcal I^{(k)}_{j_0 h}f_k'(\sum_j p_j\mathcal I^{(k)}_{jh})(1-p_i)-\beta\sum_{k}\sum_{h}\mathcal I^{(k)}_{ih}f_k(\sum_j p_j\mathcal I^{(k)}_{j h})-\delta,&\text{ }j_0=i,
    \end{cases}\\
    \Rightarrow\ &\frac{\partial g_i}{\partial p_{j_0}}|_{P=0}= \begin{cases}
    \beta \sum_{k}W^{(k)}_{i j}f_k'(0),&\text{ }j_0\neq i,\\
    \beta \sum_{k}W^{(k)}_{i j}f_k'(0)-\delta,&\text{ }j_0=i,
    \end{cases}
\end{align*}
and note that 
$$
\lambda(\beta \sum_{k}f'_k(0)W^{(k)}-\delta I)<0
\quad
\Leftrightarrow 
\quad
\lambda(\sum_{k}f'_k(0)W^{(k)})<\frac{\delta}{\beta}.
$$

\end{proof}

\subsection{Global asymptotic stability for the concave infection model}

We now show that
when $f$ is concave
the
condition in 
Theorem~\ref{thm: local asymptotic stability}
ensures global stability of the zero equilibrium, and hence guarantees that the disease dies out 
according to the mean field approximation.
\begin{defn}
Given a matrix $A$, define its symmetric version to be
$$
A^{(S)}:=(A+A^T)/2.
$$
\end{defn}
\begin{lemma}\label{lemma: spectrum inequality}
Suppose that $A$ and $B$ are $n\times n$ real matrices, and suppose that there exists a diagonal matrix $\Lambda$ such that for all $i\in \{1,2,\dots,n\},\ \Lambda_{i i}\geq 0$, and
$$
A=B-\Lambda.
$$
Then the largest eigenvalues of $A$ and $B$ satisfy $\lambda(A)\leq \lambda (B)$, and the largest eigenvalues of $A^{(S)}$ and $B^{(S)}$ also satisfy $\lambda(A^{(S)})\leq \lambda (B^{(S)})$.
\end{lemma}
\begin{proof}
Let $x$ be a unit eigenvector associated with  $\lambda(A^{(S)})$. We have
\begin{align*}
    2\lambda(A^{(S)})&=x^TAx+x^TA^Tx\\
    &=x^TBx+x^TB^Tx - 2x^T\Lambda x\\
    &=\sum_{i,j=1}^n(b_{ij}+b_{ji})x_ix_j-2\sum_{i=1}^n\Lambda_{i i} x_i^2\\
    &\leq x^T(B+B^T) x\\
    &\leq \max\{x^T(B+B^T)x\ |\ x^Tx = 1\} = 2\lambda(B^{(S)}).
\end{align*}
The inequality $\lambda(A)\leq \lambda(B)$ may be shown similarly.
\end{proof}

 \begin{thm}\label{thm: global asymptotic condition}
Suppose $f$ is concave.
If
\[
\lambda(W) \frac{f'(0) \beta}{\delta} < 1,
\] 
then $0 \in \R^n$ is a globally asymptotically stable equilibrium of 
(\ref{dynamical system})--(\ref{eq:gidef}); so
$\lim_{t\to\infty}||P(t)||=0$ 
for any valid initial condition (that is, with 
$0 \le p(0)_i \le 1$). 
\end{thm}
\begin{proof}
From the global asymptotic stability result in \cite[Lemma $1'$ ]{GSAcriterion} 
it is sufficient to show that 
all eigenvalues of the symmetric matrix $(\nabla g(P))^{(S)}$
are strictly less than $0$, for all $P\neq 0$.
We have
\begin{align*}
    \frac{\partial g_i}{\partial p_{j_0}}&=\begin{cases}
    \beta\sum_h\mathcal I_{i h}\mathcal I_{j_0 h}f'(\sum_j p_j\mathcal I_{jh})(1-p_i),&\ j_0\neq i,\\
    \beta\sum_h\mathcal I_{i h}\mathcal I_{j_0 h}f'(\sum_j p_j\mathcal I_{jh})(1-p_i)-\beta\sum_{h}\mathcal I_{ih}f(\sum_j p_j\mathcal I_{j h})-\delta,&\ j_0=i.\\
    \end{cases}\\
\end{align*}

Letting $B$ denote the $n\times n$ matrix given by
\begin{align*}
    B_{i j_0}=\begin{cases}
    \beta\sum_h\mathcal I_{i h}\mathcal I_{j_0 h}f'(\sum_j p_j\mathcal I_{jh})(1-p_i),&\ j_0\neq i,\\
    \beta\sum_h\mathcal I_{i h}\mathcal I_{j_0 h}f'(\sum_j p_j\mathcal I_{jh})(1-p_i)-\delta,&\ j_0=i,\\
    \end{cases}\\
\end{align*}
we have $\nabla g(P)= B-\Lambda$, where $\Lambda$ is the $n\times n$ diagonal matrix where for all $i\in \{1,2,\dots\}$, $\Lambda_{i i}:=\beta\sum_h \mathcal I_{i h}f(\sum_j p_j \mathcal I_{j h})\geq 0$. On the one hand Lemma~\ref{lemma: spectrum inequality} now yields $\lambda((\nabla g(P))^{(S)})\leq \lambda(B^{(S)})$; on the other hand, note that $B+\delta I\leq \nabla g(0)+\delta I$, where we interpret the inequality in a componentwise sense, and where we use $f'(\sum_j p_j \mathcal I_{j h})\leq f'(0)$,
since $f$ is concave. Hence $B^{(S)}+\delta I\leq \nabla g(0)+\delta I$, and since $B^{(S)}+\delta I$ has only non-negative entries, appealing to the Perron-Frobenius theorem, we have $\lambda(B^{(S)})\leq \lambda(\nabla g(0))$.

Combining these inequalities and using the spectral condition in the statement of the theorem, we deduce that
\[
\lambda\left(  (\nabla g(P))^{(S)} \right)
\leq
\lambda(\nabla g(0)) = \lambda \left(  \beta f'(0)W-\delta I \right) < 0,
\]
as required.
\end{proof}
A straightforward
adaptation of the proof of Theorem~\ref{thm: global asymptotic condition} yields the following global asymptotic stability result for 
the more general partitioned model.

\begin{thm}\label{thm: generalized global asymptotic stability}
Suppose all $f_k$ are concave.
If
\[
\lambda(\sum_{k = 1}^{K}f_k'(0)W^{(k)}) 
\frac{\beta}{\delta}
< 1,
\] then
$0 \in \R^n$ is a globally asymptotically stable 
equilibrium of (\ref{dynamical system}), with 
$g_i$ defined in 
(\ref{eq:mfpart}). 
\end{thm}

\section{Simulations to Test the Spectral Condition}
\label{sec:more_sim}
We now show the results of experiments that 
test the sharpness of our spectral vanishing condition.
Here, we used 
the concave functions $f(x)=2\log(1+x)$ (on the left of Figure~\ref{fig: both_i_infty vs beta} and in
Figure~\ref{fig: 2log i_t vs t}) and $f(x) = \mathrm{arctan}(x)$ (on the left in 
Figure~\ref{fig: both_i_infty vs beta} and in
Figure~\ref{fig: atan i_t vs t}) to construct partitioned models
with $f_1(x)=x$ and $f_k(x)=f(x)$ for all $k\geq 2$. We fixed a hypergraph with $n = 400$ nodes, $400$ edges, $200$ hyperedges of size $3$, $100$ hyperedges of size $4$ and $50$ hyperedges of size $5$.
At time zero, each node
was infected with independent probability 
$i_0=0.5$ 
and we used a recovery rate of 
$\delta =1$.
In addition to the mean field ODE, we also simulated the microscale model, averaged over $5$ runs, using
the discretization scheme 
described in Section~\ref{subsec:comp},
with $\Delta t=0.1$.

In Figure~\ref{fig: both_i_infty vs beta} 
the 
asterisks (red) show the corresponding
proportion of infected individuals
according to the
mean field model, 
$\sum_i p_i(t)/n$,
at time $t = 150$,
for a range of different 
$\beta$ between $0$ and $0.2$: $\beta\in \{(0.01)k\ |\ k\in \{0,\dots,20\}\}$.
The crosses (blue) show the 
corresponding proportion of infected individuals
from the microscale model,
$\sum_i X_i(t)/n$,
at time $t=150$. 
 The vertical green line represents the critical value $\beta_c:=\delta/\lambda(\sum_{k=1}^{K}f_k'(0)W^{(k)})$ (we have $\beta_c\approx 0.0265$ on the left and $\beta_c\approx 0.0490$ on the right).
 
 For the mean field model,
 we know from 
 Theorem~\ref{thm: generalized global asymptotic stability} that 
  $\beta < \beta_c$
 guarantees global stability of the 
 zero-infection state. We see that $\beta_c$ also lies close to the threshold beyond which 
 extinction 
 of the disease is lost in the mean field model. 
 For the individual-level stochastic model,
 Theorem~\ref{thm: generalized exponential decay}
 below
 shows that 
  $\beta < \beta_c$ is also 
  sufficient for eventual extinction of the disease.
  This is consistent with the results in Figure~\ref{fig: both_i_infty vs beta}.
  
\begin{figure}[htp]
    \centering
    \includegraphics[width=0.45\textwidth]{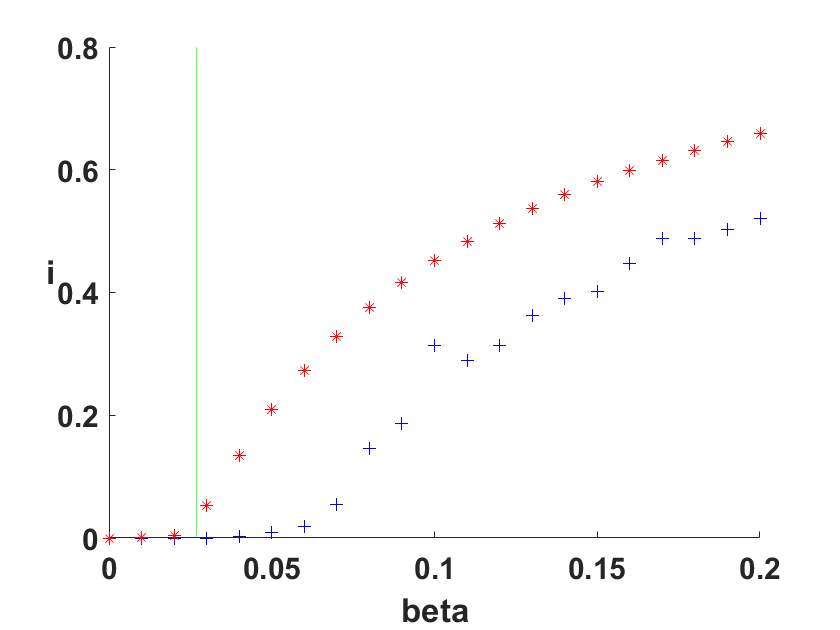}
    \includegraphics[width=0.45\textwidth]{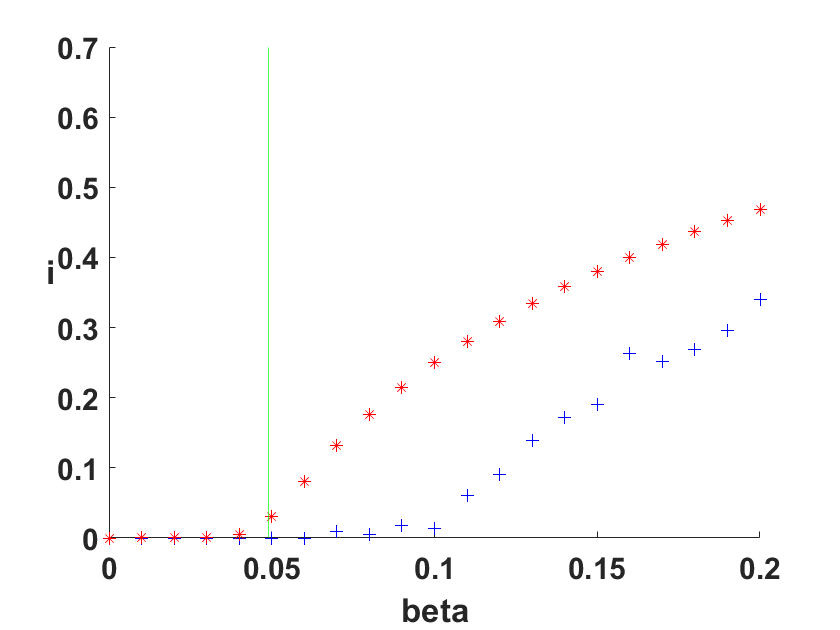}
    \caption{Left: infection function based on $2\log(1+x)$.
    Right: infection function based on $\mathrm{arctan}(x)$.
     For different choices of infection strength 
      $\beta$ (horizontal axis), we show 
      the proportion of infection individuals at time 
      $ t = 150$ (vertical axis) for the mean 
      field approximation 
   (\ref{dynamical system}) with (\ref{eq:mfpart})
   in red asterisks and 
   for the   
   individual-level stochastic model 
     (\ref{infection rate random})
     in blue crosses.
      The spectral bound arising from our analysis is show as a green vertical line.
    }
    \label{fig: both_i_infty vs beta}
\end{figure}

In the left of 
Figures~\ref{fig: 2log i_t vs t} 
and \ref{fig: atan i_t vs t} 
we show individual 
trajectories 
of the 
proportion of infected individuals,
$\sum_i p_i(t)/n$,
according to the
mean field model,
for a range of $\beta$ values.
For the same range of $\beta$ values, the plots on the right of these figures show the 
corresponding proportion of infected individuals
from the microscale model,
$\sum_i X_i(t)/n$.
 The curves are colored in red if the spectral vanishing condition $\beta < \beta_c$ is satisfied.
We see qualitative agreement between the 
mean field and individual-level models, and 
extinction for the $\beta$ values below the spectral threshold.

\begin{figure}[htp]
    \centering
    \includegraphics[width=15cm, height=7cm]{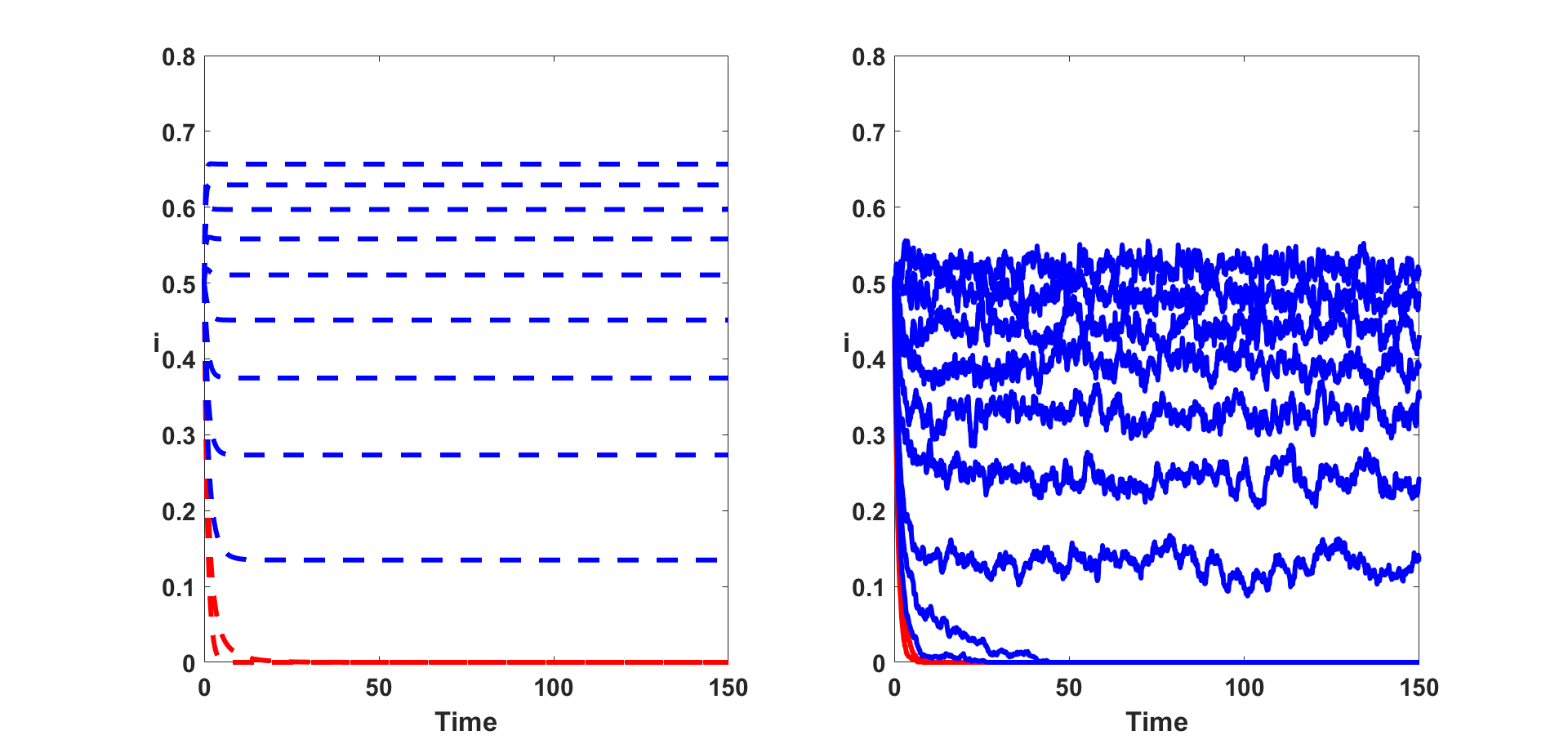}
    \caption{Results with the arctan infection function.
    Left: proportion of infected individuals using the
mean field model.
Right: proportion of infected individuals using the
individual-level model.
From bottom to top, the $\beta$ values used were
From bottom to top, the $\beta$ values used were
$\beta\in \{(0.02)k\ |\ k\in \{0,\dots,10\}\}$.
Cases where $\beta$ is below the spectral 
bound are colored in red.
    }
    \label{fig: 2log i_t vs t}
\end{figure}

\begin{figure}[htp]
    \centering
    \includegraphics[width=15cm, height=7cm]{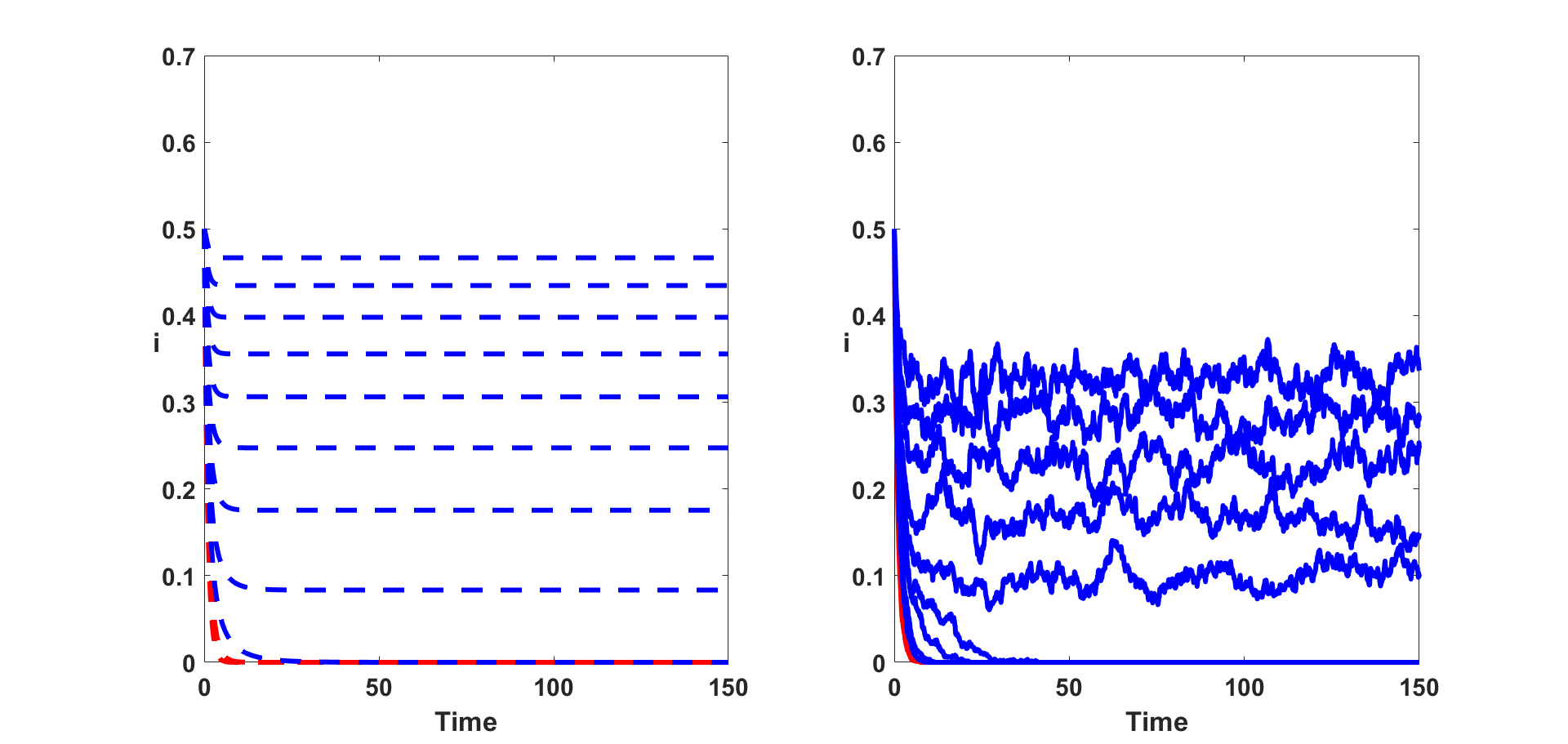}
    \caption{Results with the arctan infection function.
    Left: proportion of infected individuals using the
mean field model.
Right: proportion of infected individuals using the
individual-level model.
From bottom to top, the $\beta$ values used were
$\beta\in \{(0.02)k\ |\ k\in \{0,\dots,10\}\}$.
Cases where $\beta$ is below the spectral 
bound are colored in red.
    }
    \label{fig: atan i_t vs t}
\end{figure}

Having derived and tested spectral conditions 
that concern extinction of the disease 
at the mean field approximation level,
in the next section 
we study the microscale model directly.

\section{Exact Model}
\label{sec:exact_model}

To proceed, we recall our assumption that at time zero each node has
the same, independent, probability, $i_0$, of being infectious;
so $\mathbb{P} ( X_j(0) = 1 ) = i_0$ for
all $1 \le j \le n$.
This implies that 
$n\, i_0$ is the expected number of infectious individuals at time zero.

We are interested in the stochastic process
$\sum_i X_i(t)$, which records the number of 
infected individuals.
Our analysis generalizes arguments in  
\cite{epidemicsSpread}, which considered a stochastic 
SIS 
model on a graph
with $f$ as  
the identity map.

\subsection{Extinction}
\label{subsec:ext}

Our first result shows that the spectral condition 
arising from the mean field analysis in Theorems~\ref{thm: local asymptotic stability}
and \ref{thm: global asymptotic condition} is
also 
relevant
to the probability of extinction in the 
individual-level model.

\begin{thm}\label{thm: exponential decay}
Suppose $f$ is concave in the hypergraph infection model (\ref{infection rate random}). Then
\[
\mathbb P\left(\sum_i X_i(t)> 0\right)\leq n\, i_0\, \exp\left((\beta f'(0)\lambda(W)-\delta)t\right).
\]
Hence, if 
$
\lambda(W)
f'(0)
\beta/\delta <1
$ then the disease vanishes at an exponential rate.
\end{thm}
\begin{proof}
Consider the continuous time Markov process $\{(Y_i(t))_{t\geq 0}\}_{i=1}^{n}$ 
taking values in $\N^n$, with transition of states defined for every $1 \le i\le n$ and $\ t\geq 0$ by
$$
\begin{cases}
k&\rightarrow\ k+1\text{, with rate }\beta f'(0)\sum_j W_{i j}Y_j(t),\\
k&\rightarrow\ k-1\text{, with rate }\delta.
\end{cases}
$$
This new process is introduced here purely for the purpose of analysis. However, it may be interpreted 
as a disease model where the state of each individual is represented by a non-negative integer that indicates severity of infection. Here, exposure to 
highly infected individuals raises the chance 
of an increase in infection severity.

Suppose also that $X_i(0)=Y_i(0)$
for all $1 \le i \le n$. Since $f$ is concave,
\[
    \beta \sum_h\mathcal I_{i h}f(\sum_j X_j(t)\mathcal I_{j h}) 
    \leq 
     \beta \sum_h\mathcal I_{i h}f'(0) \sum_j X_j(t)\mathcal I_{j h} =
    \beta f'(0)\sum_{i j}W_{i j}X_j(t),
\]
from which we see that $Y_i$ stochastically dominates $X_i$. Hence
$$
\mathbb P\left(\sum_i X_i(t)>0\right)\leq \mathbb P\left(\sum_i Y_i(t)>0\right)\leq \sum_i q_i(t),
$$
where $q_i(t):=\E[Y_i(t)]$.
In terms of the 
 Chapman--Kolmogorov Equation, 
 or Chemical Master Equation,
 \cite{Gil91},
 we have 
\[
\frac{d q_i(t)}{dt}=\beta f'(0)\sum_{j}W_{i j}q_j(t)-\delta q_i(t).
\]
Letting $Q(t) = [q_1(t), q_2(t), \ldots, q_n(t)]^T$,
this linear ODE system solves to give
\[
Q(t)=\exp\left(t(\beta f'(0)W-\delta I)\right)Q(0).
\]
The matrix $\exp\left(t(\beta f'(0)W-\delta I)\right)$ is symmetric and has spectral radius 
$\exp\left((\beta f'(0)\lambda(W)-\delta I)t\right)$.
Hence, in Euclidean norm,
$$
||Q(t)|| \leq \exp((\beta f'(0)\lambda(W)-I)t)||Q(0)||. 
$$
Since $||Q(0)|| =\sqrt{n} \, i_0 $ and,
by Cauchy--Schwarz,
\[
\sum_i q_i(t)\leq \sqrt{n} ||Q(t)||,
\]
the proof is complete. 
\end{proof}
 We deduce, analogously to \cite{epidemicsSpread}, the following corollary.
\begin{corollary}
\label{cor:et}
Suppose $f$ is concave in the hypergraph infection model (\ref{infection rate random}).
 Let $\tau$ denote the time of extinction of the disease and suppose  
 $\lambda(W) f'(0) \beta <\delta$, then
 \[
 \E[\tau]\leq\frac{\log n+1}{\delta-f'(0)\beta \lambda (W)}.
 \]
\end{corollary}
\begin{proof}
Using Theorem~\ref{thm: exponential decay}, 
\begin{align*}
    \E[\tau]&=\int_0^\infty\mathbb P(\tau >t)dt\\
    &=\int_0^\infty\mathbb P(\sum_i X_i(t)>0)dt\\
    &\leq \frac{\log n}{\delta-f'(0)\beta \lambda (W)} + \int_{(\log n)/(\delta - f'(0)\beta \lambda(W))}^\infty n\exp((\beta f'(0)\lambda(W)-\delta)t)dt\\
    &\leq \frac{\log n+1}{\delta - f'(0)\beta \lambda(W)}.
\end{align*}
\end{proof}
 Likewise the partitioned case yields the following result.
\begin{thm}\label{thm: generalized exponential decay}
Suppose every $f_k$ is concave in the partitioned 
hypergraph model with infection rate 
(\ref{eq:partition}).
Then
$$
\mathbb P\left(\sum_i X_i(t)> 0\right)\leq n \, i_0 \, \exp\left(\beta \lambda(\sum_{k=1}^{K}f'_k(0)W^{(k)})-\delta\right).
$$
Hence, if 
$
\lambda(\sum_{k=1}^{K}f'_k(0)W^{(k)})
\beta /\delta < 1
$ then the disease vanishes at an exponential rate.
\end{thm}

We also have the following analogue of Corollary~7.2 on the expected time to extinction for the partitioned case.
\begin{corollary}
Suppose every $f_k$ is concave in the partitioned hypergraph model with infection rate (\ref{eq:partition}).
 Let $\tau$ denote the time of extinction of the disease and suppose  
 $\lambda(\sum_{k=1}^{K}f'_k(0)W^{(k)})
\beta /\delta < 1$, then
 \[
 \E[\tau]\leq\frac{\log n+1}{\delta-\beta \lambda (\sum_{k=1}^{K}f'_k(0)W^{(k)})}.
 \]
\end{corollary}
\subsection{Conditions that preclude extinction}
\label{subsec:hyp}

So far, we have focused on deriving thresholds 
that imply extinction. 
In this subsection, following 
ideas from \cite{epidemicsSpread}, we 
derive a condition under which 
 the disease will
persist.

Note that our analysis does not require 
the graph
associated with $W$
to be connected. 
The disconnected setting
is relevant, for example, when interventions have been imposed in order to limit interactions.
We let 
$\Delta:=D-W$ denote the Laplacian,
and let $\lambda_c(\Delta)>0$ denote the smallest non-zero eigenvalue of $\Delta$.
We also let $\emax$ denote the size of the largest hyperedge.   

\begin{defn}
Given a hypergraph $\mathcal{H}$, 
a function $f : \R_{+} \to \R_{+}$ and 
a subset of the nodes $S \subset V$, let 
$$
E(S,f):=\sum_{i\in S}\sum_{h\in E}\mathcal I_{i h} f(\sum_{j\in S^c}\mathcal I_{j h}),
$$
where $S^c:= V\backslash S$ is the complement of $S$.
Also define for integer $1\leq m\leq \floor{n/2}$
$$
\eta(\mathcal H,m,f):=\inf\left\{\frac{E(S,f)}{|S|}\ |\ \ 1\leq |S|\leq m\right\},
$$
and let $\eta(\mathcal H,m):=\eta(\mathcal H,m,Id)$.
\end{defn}
Notice that when $S$ consists of those nodes for which 
$X_i(t) = 0$,
we can write the infection transition rate of $\sum_{i=1}^{n}X_i(t)$ as $\beta E(S,f)$.
More generally, 
$\beta E(S,f)$ may be regarded as the rate at which 
nodes in the set $S$ may be infected by nodes in 
the remainder of the network. When $f=Id$ and $m=\floor{n/2}$, $\eta(\mathcal H,m)$ is the Cheeger constant, or isoperimetric number, associated with the weighted graph induced by $W=\mathcal I\mathcal I^T$. 
We may also regard $\eta(\mathcal H,m,f)$
as the smallest 
average infection rate over all subsets consisting of no more than half of the network.

The next theorem gives a probabilistic lower bound on the time to extinction. 
\begin{thm}\label{thm: non-zero steady-state}
Recall that $\tau$ denotes the hitting time of the state $0$ for the process $(\sum_j X_j(t))_{t\geq 0}$
in the hypergraph model (\ref{infection rate random}).
If $f$ is concave and $\lambda_c(\Delta)>\left(2\frac{\emax-1}{f(\emax-1)}\right)\frac{\delta}{\beta}$, then 
$$
\mathbb P\left(\tau >\frac{\floor{r^{-m+1}}}{2m}\right)\geq \frac{1-r}{e}(1+O(r^m)),
$$
where $r:=\frac{(\emax-1)\delta}{ f(\emax-1)\, \beta\,  \eta(\mathcal H,m)}<1$ and $m:=\floor{\frac{n}{2}}$.  
\end{thm}

From standard Cheeger inequalities \cite{CheegerInequalitiesSharp}, we know that the Cheeger constant of the graph induced by $W$ satisfies $2 \eta(\mathcal H, m)\geq \lambda_c(\Delta)$, hence using the assumptions on $\lambda_c(\Delta)$ in the theorem, we see that $r<1$ indeed.


In order to prove this result, we introduce the following lemma.

\begin{lemma}
If $f$ is concave and non-decreasing, then
$$\frac{f(\emax-1)}{\emax-1}\eta(\mathcal H,m)\leq \eta(\mathcal H,m,f).$$
\end{lemma}
\begin{proof}
The proof is immediate once we see that by concavity of $f$, for all $x\in \{0,\dots,\emax-1\}$
$$
\frac{f(\emax-1)}{\emax-1}x\leq f(x).
$$
\end{proof}
Now, to prove Theorem~\ref{thm: non-zero steady-state} consider the Markov process $(Z(t))_{t\geq 0}$ valued in $\{0,\dots,m\}$, with transition of states given by
$$
\begin{cases}
k\ \rightarrow\ k+1\text{ with transition rate }k\beta\frac{f(\emax-1)}{\emax-1}\eta(\mathcal H, m),\\
k\ \rightarrow\ k-1\text{ with transition rate }k\delta.
\end{cases}
$$
This Markov process is stochastically dominated by $((\sum_{i=1}^{n} X_i(t))_{t\geq 0})$, which has 
the same downward transition rate, and an upward transition rate that is at least as large:
\begin{align*}
    k\beta \frac{f(\emax-1)}{\emax-1}\eta(\mathcal H,m)&\leq k\beta\eta(\mathcal H,m,f)\\
    &\leq \beta E(S,f),
\end{align*}
where $S:=\{i\in \{1,2,\ldots,n\}\ |\ X_i(t)=0\}$. Thus, to show Theorem~\ref{thm: non-zero steady-state} it suffices to find a suitable lower bound for 
$\mathbb{P}\left(\widehat{\tau} >\frac{\floor{r^{-m+1}}}{2m}\right)$, where 
$\widehat{\tau} $ is the hitting time of $0$ for the process $(Z(t))_{t\geq 0}$. This follows by applying Theorem $4.1$ of \cite{epidemicsSpread} to the process $(Z(t))_{t\geq 0}$. 

\section{Collective Contagion Models}
\label{sec:cim}
We now consider the collective contagion models from 
\cite{APM20,simplicialSocialContagion,heterogeneityHypergraph}, where infection only starts spreading within a  hyperedge once a threshold number of infectious nodes in that hyperedge has been reached. 
As discussed in subsection~\ref{subsec:model}, collective contagion models can be represented by nonlinear functions of the form $f(x):=\max\{0,x-c\}$ for some $c>0$, or $f(x):=c_2\1(x\geq c_1)$ for some $c_1,c_2>0$.
In these cases it is obvious that the zero-infection state for the mean field approximation is locally asymptotically stable (and, indeed,
Theorem~\ref{thm: local asymptotic stability} applies).
However, 
because the functions 
are not concave, the theory found in 
Section~\ref{sec:exact_model}
for the exact model does not directly apply.
Nonetheless, we can still derive similar spectral conditions for the vanishing of the disease by finding concave functions which serve as upper bounds for $f$. For instance using $c_2\1(x\geq c_1)\leq \frac{c_2}{c_1}x\1(x\leq c_1)+c_2\1(x\geq c_1)$,
the bounds in 
Theorem~\ref{thm: exponential decay}
and Corollary~\ref{cor:et}
lead to the 
the following result.

\begin{thm}\label{thm: collective contagion model}
Suppose that $f(x):=c_2\1(x\geq c_1)$ for some $c_1,c_2>0$.
Then
$$
\mathbb P\left(\sum_i X_i(t)> 0\right)\leq n\, i_0\, \exp\left(\beta \frac{c_2}{c_1}\lambda(W)-\delta\right).
$$
In particular if $\lambda(W)<\frac{c_1}{c_2}\frac{\delta}{\beta}$, then the disease asymptotically vanishes with exponential decay and the extinction time $\tau$ satisfies
$$
\E[\tau]\leq\frac{\log n +1}{\delta -
\beta \lambda(W) c_2 / c_1}.
$$
\end{thm}

Likewise note that
$
\max\{0,x-c\}\leq \frac{\emax-1-c}{\emax-1}x,
$
where we recall that $\emax$ is the largest size of a hyperedge of $\mathcal H$. 
Hence we deduce the following result.
\begin{thm}\label{thm: max collective contagion model}
Suppose that $f(x):=\max\{0,x-c\}$ for some $c>0$. Then
$$
\mathbb P\left(\sum_i X_i(t)> 0\right)\leq n\, i_0\, \exp\left(\beta \frac{\emax-1-c}{\emax-1}\lambda(W)-\delta\right).
$$
In particular if 
\[
\lambda(W)< \left( \frac{\emax-1}{\emax-1-c} \right) \frac{\delta}{\beta},
\]then the disease asymptotically vanishes with exponential decay and the extinction time $\tau$ satisfies
\[
 \E[\tau]\leq\frac{\log n+1}{\delta-\beta \frac{\emax-1-c}{\emax-1} \lambda (W)}.
 \]
\end{thm}

\section{Discussion}
\label{sec:disc}

In this work we derived several 
 spectral conditions that control the 
 spread of disease in an SIS model on a hypergraph. 
 The conditions have the general form 
 \begin{equation}\label{new vanishing condition}
\beta \, \lambda(W) \, c_f/ \delta < 1,
\end{equation}
where $c_f>0$ is a constant depending on the
function $f$
that determines the nonlinear infection rate within a hyperedge.

We note that in the special case where 
(i) the hypergraph
is an undirected graph and hence $W$ becomes the binary adjacency matrix,
and
(ii) we have linear dependence on the number of infectious neighbors for the infection rate of a node, so $f$ is the identity function,
the condition 
(\ref{new vanishing condition})
reduces to the 
well-known vanishing spectral condition
studied in, for example, 
(\cite{epidemicsSpread,virusSpreadInNetworks,spreadingEigenvalue}).

There are two important points to be made about the 
general form of 
(\ref{new vanishing condition}).
First, the hypergraph structure appears only via the presence of the symmetric matrix 
$W \in \R^{n \times n}$.
Recall that 
$W_{ij}$ records 
the number of times that 
$i$ and $j$ both appear in the same 
hyperedge. 
Such weighted but \emph{pairwise}
information is all that feeds into this 
spectral 
threshold.
On a positive note, this implies that useful 
predictions can be made about 
disease spread on a hypergraph without 
full knowledge of the types of hyperedge 
present and the distribution of nodes within them.
(For example, when collecting human interaction data it is more reasonable to ask an individual to list each 
neighbour and state how many different ways they 
interact with that neighbour than to ask 
an individual to list all hyperedges they take part in.)
However, 
this observation also raises the possibility that 
more refined analysis might lead to 
sharper 
bounds, perhaps at the expense of simplicity and interpretability.

Our second point is that 
the new vanishing condition 
(\ref{new vanishing condition})
neatly separates three aspects:
\begin{description}
\item[(a)]
The biologically-motivated infection parameter, $\beta$.
\item[(b)]
The interaction 
structure, captured in $\lambda(W)$.
\item[(c)]
The 
coefficient $c_f$ that arises from  
modelling the nonlinear infection process.
For instance, Theorem~\ref{thm: local asymptotic stability} and Theorem~\ref{thm: global asymptotic condition} have $c_f=f'(0)$. In the collective contagion model case $f(x)=c_2\1(x\geq c_1)$, Theorem~\ref{thm: collective contagion model} indicates that we can take $c_f=c_2/c_1$.
\end{description}
We may view $\beta$ as an invariant biological constant
that reflects the underlying virulence of the disease
and is not affected by human behaviour.
The factor $\lambda(W)$, which arises from the interaction structure, will be determined by 
regional and cultural issues, including  
population density, 
age demographics,
typical household sizes, and the 
nature of prevalent commercial and manufacturing activities.
Interventions, including full or partial lockdowns, 
could be modeled through a change in $\lambda(W)$.
The third factor, $c_f$, 
is strongly dependent upon 
human behaviour and may be adjusted to reflect 
individual-based containment strategies such as social distancing, mask wearing or more frequent hand washing.
 


This work has focused on modelling, analysis and interpretation at the abstract level, concentrating on the 
fundamental question of disease extinction. 
Having developed this theory, 
it would, of course, now be of great interest to
perform  practical experiments using realistic interaction and infection data, with the aim of 
\begin{itemize}
    \item calibrating model parameters,
    \item testing hypotheses about the appropriate functional form of the infection rate,
    \item testing the predictive power of the modeling framework, especially in comparison with 
    simpler homogeneous mixing and pairwise interaction versions, 
    \item quantifying the effect of different interventions.
\end{itemize}


\bibliographystyle{siam}
\bibliography{refs}
\end{document}